\documentclass[conference]{IEEEtran}


\usepackage{graphicx,colordvi,psfrag}
\usepackage{amsmath,amssymb}
\usepackage{epstopdf}
\usepackage[caption=false]{subfig}
\usepackage{epsfig,cite}
\usepackage{calc,pstricks, pgf, xcolor}
\usepackage{nicefrac}
\usepackage{enumerate}
\usepackage{algorithm,algorithmic}
\usepackage{url}

\newtheorem{theorem}{Theorem}

\newtheorem{proposition}{Proposition}
\newtheorem{remark}{Remark}
\newtheorem{example}{Example}
\newtheorem{definition}{Definition}
\newtheorem{lemma}{Lemma}
\newenvironment{proof}[1][Proof]{\noindent\textbf{#1.} }{\ \rule{0.5em}{0.5em}}

\newcommand{\by}{\mathbf{y}}
\newcommand{\bY}{\mathbf{Y}}
\newcommand{\bx}{\mathbf{x}}
\newcommand{\bX}{\mathbf{X}}
\newcommand{\bz}{\mathbf{z}}
\newcommand{\bZ}{\mathbf{Z}}

\newcommand{\bH}{\mathbf{H}}
\newcommand{\bK}{\mathbf{K}}
\newcommand{\bA}{\mathbf{A}}
\newcommand{\ba}{\mathbf{a}}
\newcommand{\bc}{\mathbf{c}}
\newcommand{\bC}{\mathbf{C}}
\newcommand{\bd}{\mathbf{d}}
\newcommand{\bD}{\mathbf{D}}
\newcommand{\bB}{\mathbf{B}}

\newcommand{\bU}{\mathbf{U}}
\newcommand{\bu}{\mathbf{u}}
\newcommand{\bV}{\mathbf{V}}
\newcommand{\bv}{\mathbf{v}}
\newcommand{\CV}{\mathcal{V}}
\newcommand{\bw}{\mathbf{w}}
\newcommand{\bW}{\mathbf{W}}
\newcommand{\bt}{\mathbf{t}}
\newcommand{\bT}{\mathbf{T}}
\newcommand{\be}{\mathbf{e}}
\newcommand{\bE}{\mathbf{E}}
\newcommand{\bR}{\mathbf{R}}
\newcommand{\bI}{\mathbf{I}}

\newcommand{\bL}{\mathbf{L}}
\newcommand{\bg}{\mathbf{g}}
\newcommand{\bG}{\mathbf{G}}
\newcommand{\bQ}{\mathbf{Q}}
\newcommand{\bq}{\mathbf{q}}
\newcommand{\bs}{\mathbf{s}}
\newcommand{\bS}{\mathbf{S}}

\newcommand{\bF}{\mathbf{F}}
\newcommand{\bff}{\mathbf{f}}

\newcommand{\ZZ}{\mathbb{Z}}

\newcommand{\RR}{\mathbb{R}}

\newcommand{\Tsnr}{\mathsf{SNR}}

\newcommand{\Ql}{Q_{\Lambda}}
\newcommand{\Mod}{\bmod\Lambda}

\newcommand{\trace}{\mathop{\mathrm{trace}}}
\newcommand{\diag}{\mathop{\mathrm{diag}}}
\newcommand{\Span}{\mathop{\mathrm{span}}}
\newcommand{\rank}{\mathop{\mathrm{rank}}}

\newcommand{\Vol}{\mathrm{Vol}}

\newcommand{\Cwi}{C_{\text{WI}}}

\newcommand{\bZe}{\mathbf{Z}_{\text{eff}}}

\DeclareMathOperator*{\argmin}{\arg \min}

\newrgbcolor{paleblue}{.67 .87 .88}
\newrgbcolor{BoxBlue}{.9 .9 1}

\begin{document}

\title{Successive Integer-Forcing and its Sum-Rate Optimality}

\author{\IEEEauthorblockN{Or Ordentlich}
\IEEEauthorblockA{Tel Aviv University\\
ordent@eng.tau.ac.il}
\and
\IEEEauthorblockN{Uri Erez}
\IEEEauthorblockA{Tel Aviv University\\
uri@eng.tau.ac.il}
\and
\IEEEauthorblockN{Bobak Nazer}
\IEEEauthorblockA{Boston University\\
bobak@bu.edu}
\IEEEoverridecommandlockouts
\IEEEcompsocitemizethanks{
\IEEEcompsocthanksitem
The work of U. Erez was supported in part by the Israel Science Foundation under Grant No. 1557/13. The work of O. Ordentlich was supported by the Adams Fellowship Program of the Israel Academy of Sciences and Humanities, a fellowship from The Yitzhak and Chaya Weinstein Research Institute for Signal Processing at Tel Aviv University and the Feder Family Award. The work of B. Nazer was supported by NSF grant CCF-1253918.
}}

\maketitle

\begin{abstract}
Integer-forcing receivers generalize traditional linear receivers for the multiple-input multiple-output channel by decoding integer-linear combinations of the transmitted streams, rather then the streams themselves. Previous works have shown that the additional degree of freedom in choosing the integer coefficients enables this receiver to approach the performance of maximum-likelihood decoding in various scenarios. Nonetheless, even for the optimal choice of integer coefficients, the additive noise at the equalizer's output is still correlated. In this work we study a variant of integer-forcing, termed successive integer-forcing, that exploits these noise correlations to improve performance. This scheme is the integer-forcing counterpart of successive interference cancellation for traditional linear receivers. Similarly to the latter, we show that successive integer-forcing is capacity achieving when it is possible to optimize the rate allocation to the different streams. In comparison to standard successive interference cancellation receivers, the successive integer-forcing receiver offers more possibilities for capacity achieving rate tuples, and in particular, ones that are more balanced.
\end{abstract}

\section{Introduction}
\label{sec:Intro}
The integer-forcing (IF) linear receiver architecture, proposed in~\cite{zneg12IT}, provides an alternative to standard linear receivers for the Gaussian multiple-input multiple-output (MIMO) channel. Classical architectures, such as zero-forcing (ZF) and linear minimum mean-squared error (MMSE) receivers, first equalize the channel to the identity matrix $\bI$ and then decode each data stream separately via single-user decoders. While this reduces the implementation complexity (as compared to jointly decoding the data streams), it comes at the cost of a significant rate loss. This is due to the fact that, after equalization, the total noise power is spread unevenly across data streams. If there is no channel state information at the transmitter (CSIT), it is not possible to allocate rates to compensate for this noise perturbation. The main advantage of the IF receiver over classical linear receivers is that it has the freedom to equalize the channel to any full-rank integer-valued matrix $\bA$. This helps the receiver reduce the correlations between the noises experienced by each of its single-user decoders, and balance the noise power across them. Using the compute-and-forward strategy \cite{ng11IT}, each of these decoders then recovers an integer-linear combination of the data streams. Finally, the resulting noise-free linear combinations are solved for the desired data streams. We note that the complexity of the IF receiver is comparable\footnote{The additional complexity comes from computing the target integer-valued matrix $\mathbf{A}$. This is required only once per coherence interval, meaning that the added complexity is negligible for long coherence times.} to that of a classical ZF or linear MMSE receiver.

Recent work~\cite{oe13} has shown that the IF receiver can attain the capacity of the Gaussian MIMO channel to within a constant number of bits in an open-loop scenario (no CSIT), provided that an appropriate universal linear precoding operation is applied at the transmitter. Moreover, even without precoding at the transmitter, it is shown in~\cite{oe13} that for almost every channel matrix the IF receiver attains the total degrees-of-freedom (DoF) offered by the channel, even when the number of receive antennas is smaller than the number of transmit antennas and is unknown at the transmitter. This is in sharp contrast to standard linear receivers that cannot achieve any DoF in such scenarios. As an example consider the $M$-user Gaussian multiple-access channel (MAC) where each user is equipped with one transmit antenna and the receiver is also equipped with a single antenna. Obviously, applying the linear MMSE equalizer on the channel's output would result in highly suboptimal performance, as there are not enough observations to separate the transmitted signals at the receiver. With IF equalization, on the other hand, the ratios between the individual rates achieved by each user and the symmetric capacity tend to one for almost all channel gains as the signal-to-noise ratio (SNR) increases~\cite{CoFTransformFull}.

Beyond its role as a low-complexity receiver architecture, IF also has several theoretical advantages. In particular, IF equalization exploits the closure of linear/lattice codebooks w.r.t.  integer-linear combinations. In the last decade, lattice codes were found to play a key role in characterizing the fundamental limits of certain communication networks, see e.g.~\cite{pzek11,bpt10,ng11IT,wnps10,ncl10,CoFTransformFull}. A common feature of several of these lattice-based coding schemes is that, from the perspective of each receiver, they induce effective multiple-access channels with a reduced number of users, all of which employ the same lattice codebook. The achievable rates for a MAC where all users use the same lattice codebook is difficult to analyze~\cite{oe13IT}, but can be lower bounded by the rates attained via the IF receiver~\cite{CoFTransformFull}.

The performance of standard linear receivers can be improved using successive interference cancelation (SIC). The key idea is to use decoded streams to improve the channel quality for decoding subsequent streams. In this paper, we develop and analyze an analogous scheme for IF, dubbed \emph{successive integer-forcing}. That is, the receiver will use decoded {linear combinations} in order to improve the channel quality for decoding the remaining {linear combinations}. The idea of applying SIC to IF has been partially explored in the literature. In~\cite{znoeg10}, a successive decoding procedure was developed for the IF receiver in the setting where the number of receive antennas is at least as high as the number of transmit antennas. However, the optimal filter design was not found, meaning that the obtained achievable rates are suboptimal. In~\cite{nazer12IZS}, a successive procedure was developed for decoding two linear combinations over an $M$-user Gaussian MAC. The optimal filters and the highest achievable rates were found for this scenario.

The contribution of the present work is a successive IF scheme that is suitable for any number of transmit and receive antennas, and any number of desired linear combinations. Through standard linear filtering theory, we derive closed-form expressions for the optimal filters and characterize the corresponding achievable rates. We also show that the optimal integer coefficients for successive IF can be obtained using Korkin-Zolotarev lattice basis reduction. Thus, in contrast to standard IF, the optimal integer-valued matrix $\mathbf{A}$ (in terms of achievable rate) for successive IF is always unimodular.

The IF scheme is most advantageous in an open-loop scenario where the transmitter does not know the channel gains. In this case, it does not know how to allocate rates to the different streams, and therefore a single codebook is usually used for encoding each of the transmitted streams. However, as mentioned above, the IF scheme is also of theoretical interest in the context of network communication problems, where it is commonly assumed that all transmitters know all channel gains. In this case, the rates of the different streams transmitted in the IF scheme can be appropriately allocated. We show that with a judicious rate allocation, under several technical conditions, successive IF can achieve rate tuples whose sum equals the channel's multiple-access sum-capacity. While this property is shared with the standard MMSE-SIC equalizer~\cite{vg97}, successive IF equalization achieves additional sum-rate optimal rate tuples that are not attained by MMSE-SIC equalization without rate-splitting or time-sharing. These rate tuples lie closer to the symmetric capacity than those attained by the MMSE-SIC equalizer.

\section{Preliminaries}
\label{sec:Preliminaries}

In this section, we review some key results that will be useful in our derivation of the successive IF receiver. Throughout the paper, lowercase boldface variables will refer to column or row vectors (e.g., $\bx \in \mathbb{R}^{n\times 1}$ or $\bx \in \mathbb{R}^{1\times n}$). In general, we use row vectors for vectors whose entries correspond to different time indices and column vectors for vectors whose entries correspond to different spatial indices.
Uppercase boldface variables will refer to matrices (e.g., $\mathbf{X} \in \mathbb{R}^{M \times n}$). For a given matrix $\mathbf{X}$, we denote its transpose by $\mathbf{X}^T$ and its Forbenius norm as $\|\mathbf{X}\|_F$. We denote the identity matrix by $\mathbf{I}$ where the dimensions will be clear from the context. All logarithms are taken to base $2$ and rates are measured in bits per channel use.

We will focus on the real-valued Gaussian multiple-input multiple-output (MIMO) channel with $M$ transmit and $N$ receive antennas. The channel output is given by
\begin{align}
\bY=\bH\bX+\bZ\label{channelmodel}
\end{align} where $\bH \in \RR^{N\times M}$ is the channel matrix, $\bX \in \RR^{M \times n}$ is the channel input across the $M$ transmit antennas over $n$ channel uses, and $\bZ \in \RR^{N \times n}$ is additive noise that is elementwise i.i.d. Gaussian with zero mean and unit variance. The channel input is subject to the power constraint
\begin{align}
\frac{1}{n}\mathbb{E}\|\bX\|_F^2\leq M\cdot \Tsnr \ .\nonumber
\end{align}

\begin{remark} Following the steps in~\cite[Appendix C]{ng11IT}, we can show that our results also hold under a strict power constraint $\frac{1}{n} \|\bX\|_F^2 \leq M \cdot \Tsnr$.
\end{remark}

\begin{remark}
All of our results can be immediately extended to $M \times N$ complex-valued MIMO channels by expressing these channels in terms of their $2M \times 2N$ real-valued decomposition. See~\cite{zneg12IT} for more details.
\end{remark}

\subsection{MIMO Capacity}

The capacity of the MIMO channel is given by~\cite{telatar99}
\begin{align}
C=\max_{\substack{\bQ\succ 0 \\  \trace({\bQ})\leq M\cdot\Tsnr}}\frac{1}{2}\log\det\left(\bI+\bQ\bH^{T}\bH\right).\label{Capacity}
\end{align}
The choice of $\bQ$ that maximizes~\eqref{Capacity} is determined by the water-filling solution.
Often, the suboptimal choice \mbox{$\bQ=\Tsnr\cdot\bI$} is used, resulting in the \emph{white-input (WI) mutual information}
\begin{align}
\Cwi=\frac{1}{2}\log\det\left(\bI+\Tsnr \ \bH^{T}\bH\right).\nonumber
\end{align}

\begin{remark}
Note that if we restrict the $M$ rows of the channel input $\bX$ to be independent with power constraints $\frac{1}{n} \mathbb{E}\|\bx_m\|^2\leq \Tsnr$ for $m=1,\ldots,M$, the channel model~\eqref{channelmodel} describes an $M$-user Gaussian MAC, where each user is equipped with a single transmit antenna and the receiver with $N$ antennas. In this case, $\Cwi$ is the sum-capacity~\cite[Eq. (15.153)]{coverthomas}.
\end{remark}

\subsection{Successive Interference Cancelation via Noise-Prediction}
\label{subsec:SICnp}

In the sequel, we will show that successive IF can achieve $\Cwi$, provided that the transmitter can allocate rates correctly to the different streams. In order to gain some intuition, we review the standard MMSE-SIC scheme, also known as V-BLAST. However, rather than describing the receiver's decoding procedure in the common way, where it successively decodes streams and subtracts them to get a ``cleaner'' channel, we describe a receiver that performs \emph{noise prediction}. These two variants are known to be equivalent~\cite{fischer02}.

Assume each of the $M$ antennas transmits a stream $\bx_m\in\RR^{1\times n}$ of length $n$ taken from an i.i.d. Gaussian code with average power $\Tsnr$, independent of the streams transmitted by the other antennas. The receiver decodes the transmitted streams in $M$ successive steps, where in each step a single stream is decoded. The receiver first performs linear MMSE estimation of $\bX=[\bx_1^T \ \cdots \ \bx_M^T]^T$ from $\bY$ using the filter matrix
\begin{align}
\bB={\bH}^T\left(\frac{1}{\Tsnr}\bI+{\bH}{\bH}^T\right)^{-1}.\label{LMMSEmat}
\end{align}
This gives rise to the effective channel
\begin{align}
\bY_{\text{eff}}=\bB\bY=\bX+\bE,\nonumber
\end{align}
where
\begin{align}
\bE=(\bB\bH-\bI)\bX+\bB\bZ,\nonumber
\end{align}
is the estimation error. %Note that $\bE$ and $\bX$ are correlated, and therefore the effective channel is not an additive noise one. However, by appropriately scaling each row of $\hat{\bX}$ this correlation can be removed, and it is shown in~\cite[Lemma 2]{cdef95} that this scaling procedure reduces the effective SNR by $-1$.
Since $\bX$ and $\bZ$ are statistically independent and all their entries are i.i.d. Gaussian, the columns of $\bE$ are Gaussian vectors, statistically independent of each other. Each column of $\bE$, corresponding to a different time index, is a zero-mean Gaussian vector with covariance matrix
\begin{align}
\bK_{\be\be}&=\Tsnr(\bB\bH-\bI)(\bB\bH-\bI)^T+\bB\bB^T\label{LMMSEtmp}\\
&=\Tsnr(\bI+\Tsnr\ \bH^T\bH)^{-1},\label{LMMSE}
\end{align}
where~\eqref{LMMSE} follows by substituting~\eqref{LMMSEmat} into~\eqref{LMMSEtmp} and applying Woodbury's matrix identity (i.e., the Matrix Inversion Lemma)~\cite[Thm 18.2.8]{harville}.

The matrix $(\bI+\Tsnr\ \bH^T\bH)^{-1}$ is symmetric and positive definite, and therefore admits a (unique) Cholesky decomposition
\begin{align}
\left(\bI+\Tsnr\ \bH^T\bH\right)^{-1}=\bG\bG^T,\label{cholnoA}
\end{align}
where $\bG\in\RR^{M\times M}$ is a lower triangular matrix with strictly positive diagonal entries. It follows that
\begin{align}
\bK_{\be\be}=\Tsnr \ \bG\bG^T,
\end{align}
and $\bE$ can be written as $\bE=\sqrt{\Tsnr}\ \bG\bW$, where $\bW$ is an \mbox{$M\times n$} matrix of i.i.d. Gaussian entries with zero mean and unit variance. Thus, the effective channel can be written as
\begin{align}
\bY_{\text{eff}}=\bX+\sqrt{\Tsnr}\ \bG\bW.\label{effchan_np}
\end{align}
Now, the receiver starts successively decoding the streams. First, it uses $\by_{\text{eff},1}$, the first row of $\bY_{\text{eff}}$, to decode the first stream $\bx_1$. Denoting the $(i,j)$-th entry of $\bG$ by $g_{ij}$, this can be done as long as the rate of this stream satisfies
\begin{align}
R_1&<\frac{1}{2}\log\left(1+\frac{\Tsnr}{g_{11}^2\Tsnr}-1 \right)\nonumber\\
&=-\frac{1}{2}\log\left(g_{11}^2\right).\nonumber
\end{align}
where the $-1$ term inside the logarithm compensates for the fact that $\bE$ and $\bX$ are correlated, as explained in~\cite[Lemma 2]{cdef95}.
After correctly decoding $\bx_1$, the receiver can obtain $\bw_1$, the first row of $\bW$, as
\begin{align}
\bw_1=\frac{\by_{\text{eff},1}-\bx_1}{\sqrt{\Tsnr}g_{11}} \ ,
%\bw_1=(\by_{\text{eff},1}-\bx_1)/(\sqrt{\Tsnr}g_{11}),
\end{align}
and produce a less noisy channel
\begin{align}
\bY_{\text{eff}}^{(2)}&=\bY_{\text{eff}}-\sqrt{\Tsnr}\bg_1\bw_1\nonumber\\
&=\bX+\sqrt{\Tsnr}\cdot\bG^{(2)}\bW,\nonumber
\end{align}
where $\mathbf{g}_1$ is the first column of $\bG$ and \mbox{$\bG^{(2)}=\bG-\mathbf{g}_1$} is the matrix $\bG$ with its first column nulled out. Now, the receiver can decode $\bx_2$ from $\by_{\text{eff},2}^{(2)}$, the second row of $\bY_{\text{eff}}^{(2)}$, as long as its rate satisfies
\begin{align}
R_2<-\frac{1}{2}\log\left(g_{22}^2\right).\nonumber
\end{align}
Continuing in the same manner, it follows that each stream can be decoded reliably as long as
\begin{align}
R_m<-\frac{1}{2}\log\left(g_{mm}^2\right), \ \ \ m=1,\cdots,M.\nonumber
\end{align}
The described noise prediction scheme can therefore achieve the sum-rate
\begin{align}
\sum_{m=1}^M R_m&=-\frac{1}{2}\sum_{m=1}^M \log\left(g_{mm}^2\right)\nonumber\\
&=-\frac{1}{2}\log\left(\prod_{m=1}^M g_{mm}^2 \right)\nonumber\\
&=-\frac{1}{2}\log\det\left(\bG\bG^T\right)\nonumber\\
&=\frac{1}{2}\log\det\left(\bI+\Tsnr\bH^T\bH\right)\nonumber\\
&=\Cwi.\nonumber
\end{align}
In the sequel, we will see that a similar noise prediction scheme enables successive integer-forcing to achieve a sum-rate of $\Cwi$. Rather than performing linear MMSE estimation for $\bX$, in successive integer-forcing one estimates the linear combinations $\bA\bX$, and successively predicts the associated estimation errors.

\subsection{Integer-Forcing}
\label{subsec:IF}

IF equalization is a low-complexity architecture for the MIMO channel, which was proposed by Zhan \textit{et al}.~\cite{zneg12IT}. The key idea underlying IF is to first decode integer-linear combinations of the signals transmitted by all antennas, and then, after the noise is removed, invert those linear combinations to recover the individual transmitted signals. This is made possible by transmitting codewords from the \emph{same} linear\slash lattice code from all $M$ transmit antennas, leveraging the property that linear codes are closed under (modulo) linear combinations with integer-valued coefficients.

We briefly recall the IF scheme. We begin by presenting several lattice definitions. A lattice $\Lambda$ is a discrete subgroup of $\RR^n$ which is closed under reflection and real addition. We denote the nearest neighbor quantizer associated with the lattice $\Lambda$ by
\begin{align}
\Ql(\bx)=\argmin_{\mathbf{t}\in\Lambda}\|\bx-\mathbf{t}\|.
\label{NNquantizer}
\end{align}
The basic Voronoi region of $\Lambda$, denoted by $\CV$, is the set of all points in $\RR^n$ which are quantized to the zero vector, where ties in~\eqref{NNquantizer} are broken in a systematic manner. The modulo operation returns the quantization error w.r.t. the lattice,
\begin{align}
\left[\bx\right]\Mod=\bx-\Ql(\bx).\nonumber
\end{align}
and the second moment of $\Lambda$ is defined as
\begin{align}
\sigma^2(\Lambda)=\frac{1}{n}\frac{1}{\Vol(\CV)}\int_{\bu\in\CV}\|\bu\|^2d\bu,\nonumber
\end{align}
where $\Vol(\CV)$ is the volume of $\CV$.
A lattice $\Lambda$ is said to be nested in $\Lambda_1$ if $\Lambda\subseteq\Lambda_1$.
The coding scheme presented in this paper utilizes a pair of $n$-dimensional nested lattices $\Lambda_c\subset\Lambda_f$, where $\Lambda_c$ is referred to as the coarse lattice and $\Lambda_f$ as the fine lattice. A nested lattice codebook $\mathcal{C}=\Lambda_f\cap\CV_c$, with rate
\begin{align}
R=\frac{1}{n}\log\left|\Lambda_f\cap\CV_c\right| \nonumber
\end{align}
is associated with the nested lattice pair. The codebook is scaled such that $\sigma^2(\Lambda_c)=\Tsnr$.
In Section~\ref{subsec:sumrateopt} we extend the proposed coding scheme to one that uses a chain of $M+1$ nested lattices
\begin{align}
\Lambda_c\subseteq\Lambda_{f_M}\subseteq\cdots\subseteq\Lambda_{f_1},\label{nestedchain}
\end{align}
from which we construct $M$ nested codebooks.

In the IF scheme, the information bits to be transmitted are partitioned into $M$ streams. Each of these streams is encoded by the nested lattice code $\mathcal{C}$, producing $M$ row vectors, each in $\mathcal{C}\subset\RR^{1\times n}$. In particular, the $m$th stream, consisting of $nR$ information bits, is mapped to a lattice point $\bt_{m}\in\mathcal{C}$. Then, a random dither\footnote{These random dithers can be replaced with deterministic dithers without affecting the achievable rate region, i.e., no common randomness is necessary. See \cite[Appendix C]{ng11IT} for more details.} $\bd_{m}\in\RR^{1\times n}$ uniformly distributed over $\CV_c$ and statistically independent of $\bt_{m}$, known to both the transmitter and the receiver, is used to produce the signal
\begin{align}
\bx_{m}=\left[\bt_{m}-\bd_{m}\right]\Mod_c\nonumber.
\end{align}
The signal $\bx_{m}$ is uniformly distributed over $\CV_c$ and is statistically independent of $\bt_{m}$ due to the Crypto Lemma~\cite[Lemma 1]{ez04}. It follows that
\begin{align}
\frac{1}{n}\mathbb{E}\|\bx_{m}\|^2=\sigma^2(\Lambda_c)=\Tsnr.\nonumber
\end{align}
The $m$th antenna transmits the signal $\bx_{m}\in\RR^{1\times n}$ over $n$ consecutive channel uses. %Thus, the total transmission rate is $R_{\text{IF}}=MR$ bits/channel use.

% The received signal is
%\begin{align}
%\bY=\bH\bX+\bZ,\nonumber
%\end{align}
%where $\bZ\in\RR^{N\times n}$ is a matrix of i.i.d. entries with zero mean and unit variance.

Define \mbox{$\bT\triangleq[\bt_1^T \ \cdots \ \bt_M^T]^T$} to be the $M\times n$ matrix whose rows consist of the lattice points corresponding to the $M$ data streams, \mbox{$\bD\triangleq[\bd_1^T \ \cdots \ \bd_M^T]^T$} be the $M\times n$ matrix whose rows correspond to the $M$ dither vectors, and $\bX\triangleq[\bx_1^T \ \cdots \ \bx_M^T]^T$ be the matrix whose rows correspond to the $M$ channel input vectors. These inputs vectors are transmitted into the channel~\eqref{channelmodel} to yield the $N \times n$ output $\bY$.

The IF receiver chooses an equalizing filter matrix $\bB\in\RR^{M\times N}$ and a full-rank target integer-valued matrix $\bA\in\ZZ^{M\times M}$, and computes
\begin{align}
\bY_{\text{eff}}&=\left[\bB\bY+\bA\bD\right]\Mod_c\nonumber\\
&=\left[\bA{\bX}+\bA{\bD}+(\bB{\bH}-\bA){\bX}+\bB{\bZ}\right]\Mod_c\nonumber\\
&=\left[\bA{\bT}+(\bB{\bH}-\bA){\bX}+\bB{\bZ}\right]\Mod_c\nonumber\\
&=\left[\bV+\bZe\right]\Mod_c,\label{eqoutput}
\end{align}
where
\begin{align}
\bV\triangleq\left[\bA{\bT}\right]\Mod_c
\label{modeq}
\end{align}
is an $M\times n$ real-valued matrix with each row being a codeword in $\mathcal{C}$ owing to the linearity of the code,
\begin{align}
\bZe\triangleq (\bB{\bH}-\bA){\bX}+\bB{\bZ}\label{Bze}
\end{align}
is additive noise statistically independent of $\bV$ (since ${\bX}$ and ${\bZ}$ are statistically independent of ${\bT}$), and the notation $\mod \Lambda_c$ is to be understood as reducing \emph{each row} of the obtained matrix modulo the coarse lattice.
Each row of ${\bY}_{\text{eff}}$, denoted by $\by_{\text{eff},m}$, $m=1,\ldots,M$, is the modulo sum of a codeword and effective noise. Thus, the IF receiver transforms the original MIMO channel into a set of $M$ point-to-point modulo-additive sub-channels
\begin{align}
\by_{\text{eff},m}=\left[\bv_m+\bz_{\text{eff},m}\right]\Mod_c, \ \ m=1,\ldots,M.\label{effsubchannel}
\end{align}
The IF receiver decodes the output of each sub-channel separately. If the decoding is successful over all $M$ sub-channels, the receiver has access to $\bV=[\bv_1^T \ \cdots \ \bv_M^T]^T$, from which it can recover the matrix ${\bT}$ by solving the (modulo) set of equations.

Define the effective variance of $\bz_{\text{eff},m}$ as
\begin{align}
\sigma^2_{\text{eff},m}&\triangleq\frac{1}{n}\mathbb{E}\left\|\bz_{\text{eff},m}\right\|^2.\nonumber
\end{align}
It follows from~\cite{ng11IT,ez04} that for a ``good'' (capacity-achieving) nested lattice code $\mathcal{C}$ the integer-linear combination $\bv_m$ can be reliably decoded from $\by_{\text{eff},m}$ as long as
\begin{align}
R<\frac{1}{2}\log\left(\frac{\Tsnr}{\sigma^2_{\text{eff},m}}\right),\nonumber
\end{align}
and all $M$ equations can be decoded reliably if
\begin{align}
R<\min_{m=1,\ldots,M} \frac{1}{2} \log\left(\frac{\Tsnr}{\sigma^2_{\text{eff},m}}\right).\nonumber
\end{align}
Note that the additive noise vectors $\bz_{\text{eff},1},\ldots,\bz_{\text{eff},M}$ are not statistically independent. Thus, treating the $M$ sub-channels as parallel is suboptimal, and some improvement can be obtained by exploiting this coupling. In the next section, we will show how successive IF exploits the aforementioned noise correlations to enhance performance, with only a slight increase in the decoding complexity, i.e., the receiver still decodes the linear combinations one-by-one.

\subsection{Linear MMSE Estimation and Generalizations to Matrix Estimation}
\label{subsec:LMMSE}
The derivation of the optimal filters for successive IF involves several results from linear MMSE estimation.

Consider a random vector $\bx\in\RR^{M\times 1}$ with zero mean and covariance matrix \mbox{$\left\{\bK_{\bx\bx} \right\}_{ij}=\mathbb{E}(x_i x_j)$} and a random vector of measurements $\by\in\RR^{N\times 1}$ with zero mean and covariance matrix \mbox{$\left\{\bK_{\by\by} \right\}_{ij}=\mathbb{E}(y_i y_j)$}. The cross-covariance matrix between $\bx$ and $\by$ is given by \mbox{$\left\{\bK_{\bx\by} \right\}_{ij}=\mathbb{E}(x_i y_j)$}. The class of linear estimators for $\bx$ from $\by$ consists of all estimators of the form $\hat{\bx}=\bB\by$, where $\bB\in\RR^{M\times N}$. The estimation error is defined as $\be=\bx-\hat{\bx}$, and the linear MMSE criterion corresponds to minimizing $\mathbb{E}(e_i^2)$ for all $i=1,\ldots,M$ over all filters $\bB\in\RR^{M\times N}$. It is well known that the optimal estimation filter under this criterion must satisfy the orthogonality principle
\begin{align}
\mathbf{0}=\mathbb{E}\left(\be\by^T\right)=\mathbb{E}\left((\bx-\bB\by)\by^T\right),\nonumber
\end{align}
where $\mathbf{0}$ is a matrix of zeros with appropriate dimensions, and is given by
\begin{align}
\bB^*=\bK_{\bx\by}\bK_{\by\by}^{-1}.\nonumber
\end{align}
For the optimal estimator, the estimation error covariance matrix is given by
\begin{align}
\bK_{\be\be}\triangleq\mathbb{E}(\be\be^T)=\bK_{\bx\bx}-\bK_{\hat{\bx}\hat{\bx}}.\nonumber
\end{align}

In the previous subsection we have seen that the performance of the IF receiver is dictated by $\bZ_{\text{eff}}$, which can be thought of as the estimation error of $\bA\bX$ from $\bY=\bH\bX+\bZ$, when the filter $\bB$ is used. The achievable rate for IF over the $m$th sub-channel is maximized when $\sigma^2_{\text{eff},m}=\nicefrac{1}{n}\mathbb{E}\|\bz_{\text{eff},m}\|^2$ is minimized. Thus, $\bB$ should be chosen such as to minimize $\sigma^2_{\text{eff},m}$ for all $m=1,\ldots,M$. This criterion is similar to the MMSE criterion, except for the fact that here the goal is to minimize the effective variance of a (non-i.i.d.) vector, rather than the variance of a random variable. However, as we now show, the two problems are equivalent if we replace the covariance matrices of random vectors, whose entries correspond to the correlations between random variables, with \emph{generalized covariance matrices} whose entries correspond to the effective correlations between random vectors.

\begin{definition}
\label{def:gencov}
For a random matrix $\bX\in\RR^{M\times n}$ with rows $\bx_i^T \in\RR^{1\times n}$, $i=1,\ldots,M$, we define the generalized covariance matrix as
\begin{align}
\left\{\mathbf{\tilde{K}}_{\bX\bX}\right\}_{ij}\triangleq\frac{1}{n}\mathbb{E}\left(\bx_{i}^T \bx_j \right).\nonumber
\end{align}
If $\bY\in\RR^{N\times n}$ with rows $\by_j^T \in\RR^{1\times n}$, $j=1,\ldots,N$, is another random matrix, we define the generalized cross-covariance matrix of $\bX$ and $\bY$ as
\begin{align}
\left\{\mathbf{\tilde{K}}_{\bX\bY}\right\}_{ij}\triangleq\frac{1}{n}\mathbb{E}\left(\bx_{i}^T \by_j \right).\nonumber
\end{align}
\end{definition}

\vspace{1mm}

\begin{proposition}
\label{prop:gencovcalc}
Let $\bX\in\RR^{M\times n}$ and $\bY\in\RR^{N\times n}$ be two random matrices with generalized covariance matrices $\mathbf{\tilde{K}}_{\bX\bX}$ and $\mathbf{\tilde{K}}_{\bY\bY}$, respectively, and cross-covariance matrix $\mathbf{\tilde{K}}_{\bX\bY}$. Let $\bW=\bG\bX$ and $\bU=\bH\bY$ for two deterministic matrices $\bG\in\RR^{K\times M}$ and $\bH\in\RR^{L\times N}$. Then, $\mathbf{\tilde{K}}_{\bW\bU}=\bG\mathbf{\tilde{K}}_{\bX\bY}\bH^T$, and in particular, $\mathbf{\tilde{K}}_{\bW\bW}=\bG\mathbf{\tilde{K}}_{\bX\bX}\bG^T$ and $\mathbf{\tilde{K}}_{\bU\bU}=\bH\mathbf{\tilde{K}}_{\bY\bY}\bH^T$.
\end{proposition}

\begin{proof}
\begin{align}
\left\{\mathbf{\tilde{K}}_{\bW\bU}\right\}_{ij}&=\frac{1}{n}\mathbb{E}\left(\bw_i^T \bu_j \right)\nonumber\\
&=\frac{1}{n}\mathbb{E}\left(\sum_{k=1}^K\sum_{\ell=1}^L g_{ik}\bx_k^T \by_\ell h_{j \ell} \right)\nonumber\\
&=\sum_{k=1}^K\sum_{\ell=1}^L g_{ik}\left\{\mathbf{\tilde{K}}_{\bX\bY}\right\}_{k\ell}h_{j \ell} \nonumber\\
&=\left\{\bG\mathbf{\tilde{K}}_{\bX\bY}\bH^T\right\}_{ij}.\nonumber
\end{align}
\end{proof}

\begin{lemma}
\label{lem:MMSE}
Let $\bx\in\RR^{M\times 1}$ and $\by\in\RR^{N\times 1}$ be two random vectors with zero mean, covariances $\bK_{\bx\bx}$ and $\bK_{\by\by}$ and cross-covariance matrix $\bK_{\bx\by}$. Let $\bX\in\RR^{M\times n}$ and $\bY\in\RR^{N\times n}$ be two random matrices with zero mean and generalized covariance and cross-covariances as $\bx$ and $\by$, i.e., $\mathbf{\tilde{K}}_{\bX\bX}=\bK_{\bx\bx}$, $\mathbf{\tilde{K}}_{\bY\bY}=\bK_{\by\by}$ and $\mathbf{\tilde{K}}_{\bX\bY}=\bK_{\bx\by}$. Then, for any filter $\bB\in\RR^{M\times N}$
\begin{align}
\mathbf{\tilde{K}}_{\bE\bE}=\bK_{\be\be},\nonumber
\end{align}
where $\bE=\bB\bY-\bX$ and $\be=\bB\by-\bx$. In particular, the linear MMSE estimator for $x_m$ from $\by$ also minimizes the effective variance of the estimation error vector of $\bx_m$ from $\bY$ for all $m=1,\ldots,M$.
\end{lemma}
\begin{proof}
Follows immediately from Definition~\ref{def:gencov} and Proposition~\ref{prop:gencovcalc}.
\end{proof}
\section{Successive Integer-Forcing}
\label{sec:SIF}

We now describe and analyze the successive integer-forcing receiver, which combines ideas from classical successive interference cancellation and integer-forcing. At a high level, the goal of the successive IF receiver is the same as that of the IF receiver: first recover a set of integer-linear combinations described the coefficient matrix $\mathbf{A} \in \mathbb{Z}^{M \times M}$ and then solve for the desired messages. However, rather than decoding these integer-linear combinations in parallel, the successive IF receiver decodes them one at a time and uses decoded combinations to reduce the effective noise encountered in subsequent decoding steps (as in SIC).

The receiver begins by performing linear MMSE estimation of $\bA\bX$ from $\bY$, adding back the dithers and reducing the result modulo the coarse lattice, just as in standard IF. The resulting effective channel is given by~\eqref{eqoutput},~\eqref{modeq}, and~\eqref{Bze} where $\bB$ is chosen as
\begin{align}
\bB=\bA{\bH}^T\left(\frac{1}{\Tsnr}\bI+{\bH}{\bH}^T\right)^{-1}.\nonumber
\end{align}
The resulting generalized covariance matrix of $\bZ_{\text{eff}}$ is
\begin{align}
\mathbf{\tilde{K}}_{\bZ_{\text{eff}}\bZ_{\text{eff}}}=\Tsnr \  \bA(\bI+\Tsnr\ \bH^T\bH)^{-1}\bA^T.\nonumber
\end{align}
For a full-rank matrix $\bA\in\ZZ^{M\times M}$, the matrix $\bA(\bI+\Tsnr \ \bH^T\bH)^{-1}\bA^T$ admits a Cholesky decomposition
\begin{align}
\bA\left(\bI+\Tsnr \ \bH^T\bH\right)^{-1}\bA^T=\bL\bL^T,\label{cholnp}
\end{align}
where $\bL\in\RR^{M\times M}$ is a lower triangular matrix with strictly positive diagonal entries. Let $$\bW=\frac{1}{\sqrt{\Tsnr}}\bL^{-1}\bZ_{\text{eff}}$$ and note that $\mathbf{\tilde{K}}_{\bW\bW}=\bI$, by Proposition~\ref{prop:gencovcalc}. Now, the effective channel is
\begin{align}
\bY_{\text{eff}}=\left[\bV+\sqrt{\Tsnr}\ \bL\bW\right]\Mod_c,\label{zwhitened}
\end{align}
where $\bW$ is statistically independent of $\bV$ as it is a deterministic function of $\bZ_{\text{eff}}$ which is statistically independent of $\bV$.

As in Section~\ref{subsec:SICnp}, the receiver successively decodes the integer-linear combinations one-by-one. After decoding the $m$th combinations $\bv_m$, it recovers $\bw_m$ and cancels its contribution to the effective noises that corrupt equations that are yet to be decoded. Specifically, the receiver begins by decoding $\bv_1$ from $\by_{\text{eff},1}$. This can be done reliably if
\begin{align}
R<\frac{1}{2}\log\left(\frac{\Tsnr}{\Tsnr \ \ell_{11}^2}\right)=-\frac{1}{2}\log(\ell_{11}^2).\label{ratev1}
\end{align}
Assuming $\bv_1$ was decoded correctly, the receiver next computes
\begin{align}
\mathbf{\hat{w}}_1&=\frac{1}{\sqrt{\Tsnr}\ell_{11}}\left[\by_{\text{eff},1}-\bv_1 \right]\Mod_c\nonumber\\
&=\frac{1}{\sqrt{\Tsnr}\ell_{11}}\left[\sqrt{\Tsnr}\ \ell_{11}\bw_1 \right]\Mod_c\nonumber\\
&\stackrel{w.h.p.}{=}\bw_1,\label{w1hat}
\end{align}
where~\eqref{w1hat} follows from the fact that for a ``good'' nested lattice codebook~\eqref{ratev1} implies that $\sqrt{\Tsnr}\ell_{11}\bw_1\in\CV_f\subset\CV_c$ with high probability. See Section~\ref{subsec:correctzeff} for a comprehensive discussion on this assumption. The receiver then uses $\bw_1$ to produce a less noisy channel
\begin{align}
\bY_{\text{eff}}^{(2)}&=\left[\bY_{\text{eff}}-\sqrt{\Tsnr}\ \mathbf{l}_1\bw_1\right]\Mod_c\nonumber\\
&=\left[\bV+\sqrt{\Tsnr}\ \bL^{(2)}\bW\right]\Mod_c,\nonumber
\end{align}
where $\mathbf{l}_1$ is the first column of $\bL$ and \mbox{$\bL^{(2)}=\bL-\mathbf{l}_1$} is the matrix $\bL$ with its first column nulled out. Now, the receiver can decode $\bv_2$ from $\by_{\text{eff},2}^{(2)}$, the second row of $\bY_{\text{eff}}^{(2)}$, as long as
\begin{align}
R<-\frac{1}{2}\log\left(\ell_{22}^2\right).\nonumber
\end{align}
Continuing in the same manner, it follows that all equations can be decoded reliably as long as
\begin{align}
R<-\frac{1}{2}\log\left(\max_{m=1,\ldots,M}\ell_{mm}^2\right).\label{MACrate}
\end{align}
This is summarized in the following theorem.

\vspace{1mm}

\begin{theorem}
\label{thm:SIC}
There exist nested lattice codebooks such that, for any full-rank matrix $\bA\in\ZZ^{M\times M}$, successive IF can achieve any rate satisfying
\begin{align}
R_{\text{S-IF}}<-\frac{M}{2}\log\left(\max_{m=1,\ldots,M}\ell_{mm}^2\right),\nonumber
\end{align}
where $\ell_{mm}$ are the diagonal entries of $\bL$ from~\eqref{cholnp}.
\end{theorem}

\vspace{1mm}

Note that in the described scheme each antenna transmits an independent stream. Thus, the same coding scheme can be applied over an $M$-user Gaussian MAC, where the transmit antennas are distributed. Hence, for a Gaussian $M$-user MAC, with the $m$th column of $\bH$ representing the coefficients from the $m$th user to the receiver, each user can achieve any rate satisfying~\eqref{MACrate}.

\subsection{Sum Rate Optimality of Successive IF}
\label{subsec:sumrateopt}

In this subsection, we consider using a \emph{chain} of nested lattice codebooks. Specifically, $M$ nested lattice codebooks $\mathcal{C}_m=\Lambda_{f_m}\cap\CV_c$, $m=1,\ldots,M$, are constructed from the lattice chain~\eqref{nestedchain}. Note that $\mathcal{C}_M\subseteq\cdots\subseteq\mathcal{C}_1$
by construction, and the associated rates satisfy $R_M\leq\cdots\leq R_1$. Each of the $M$ streams is encoded by one of these codebooks. We show that with such a chain of nested lattice codebooks, successive IF can achieve $\Cwi$ if the transmitter judiciously allocates the rates to the different streams, and if the diagonal entries of $\bL$ from~\eqref{cholnp} are monotonically increasing.

The main idea is that each integer-linear combination $\bv_m$ can be decoded reliably if it is taken from a good nested lattice codebook of rate smaller than $-\nicefrac{1}{2}\log(\ell_{mm}^2)$. Thus, if we could ensure that $\bv_m$ belongs to the codebook $\mathcal{C}_m$, for all $m=1,\ldots,M$, we could just choose the rates of the $M$ nested codebooks to satisfy $R_m<-\nicefrac{1}{2}\log(\ell_{mm}^2)$.

%For this reason, in order to ensure correct decoding of all equations, regardless of the integer coefficients $\{a_{mk}\}$, all codebooks should be of rates smaller than $-\nicefrac{1}{2}\log(\max\ell_{mm}^2)$.

However, $\bv_m=[\sum_{k=1}^M a_{mk}\bt_k]\Mod_c$ belongs to the densest lattice codebook from which the codewords $\bt_1,\ldots,\bt_M$ are taken. This obstacle can be overcome by using the equations that were already decoded not only for estimating the noises corrupting the remaining equations, but also for \emph{reducing the rates} of the remaining equations. This is essentially done by adding integer multiples of decoded equations to the remaining ones in a way that nulls out the effect of some of the lattice points $\bt_k$ participating in these equations.

The idea of using decoded equations for reducing the rates of the remaining ones was proposed and analyzed in~\cite[Section IV.B]{CoFTransformFull}. For sake of brevity, we do not repeat the details, and only briefly illustrate the idea by a simple example. Assume that the number of transmit antennas is $M=2$ and $\mathcal{C}_2\subset\mathcal{C}_1$ are two nested linear codes over the prime field $\ZZ_p$ with rates $R_2\leq R_1$. Two nested lattice codebooks are constructed by mapping $\mathcal{C}_1$ and $\mathcal{C}_2$ to a $p$-PAM constellation, and the coarse lattice, in this case, is taken as $\Lambda_c=p\ZZ^n$. The first antenna transmits a codeword $\bt_1\in\mathcal{C}_1$ and the second transmits $\bt_2\in\mathcal{C}_2$. The effective channel~\eqref{eqoutput} is
\begin{align}
\by_{\text{eff},1}&=\left[a_{11}\bt_1+a_{12}\bt_2+\bz_{\text{eff},1}\right] \hspace{-0.13in}\mod p\nonumber\\
\by_{\text{eff},2}&=\left[a_{21}\bt_1+a_{22}\bt_2+\bz_{\text{eff},2}\right] \hspace{-0.13in}\mod p\nonumber,
\end{align}
where we assume w.l.o.g. that $\sigma^2_{\text{eff},1}\leq\sigma^2_{\text{eff},2}$. The first equation $\bv_1=[a_{11}\bt_1+a_{12}\bt_2] \hspace{-0.08in}\mod p$ is a codeword in $\mathcal{C}_1$ and can be decoded if $R_1$ is small enough w.r.t. $\sigma^2_{\text{eff},1}$. After decoding $\bv_1$ the receiver can scale it by $a_{21}a_{11}^{-1}$, where the inversion is over the field $\ZZ_p$, and subtract it from $\by_{\text{eff},2}$ to obtain
\begin{align}
\tilde{\by}_{\text{eff},2}&=\left[a_{21}\bt_1+a_{22}\bt_2-a_{21}a_{11}^{-1}\bv_1+\bz_{\text{eff},2}\right]\hspace{-0.13in}\mod p\nonumber\\
&=\left[(a_{22}-a_{21}a_{11}^{-1}a_{12})\bt_2+\bz_{\text{eff},2}\right]\hspace{-0.13in}\mod p\nonumber.
\end{align}
Now, ${\bv}_2^{(2)}\triangleq[(a_{22}-a_{21}a_{11}^{-1}a_{12})\bt_2] \hspace{-0.08in}\mod p$ is in $\mathcal{C}_2$ and it suffices that $R_2$ is small enough w.r.t. $\sigma^2_{\text{eff},2}$ to ensure correct decoding. Thus, the described procedure enables to ``allocate'' different rates to the different equations.

It was shown in~\cite{CoFTransformFull} that for any full-rank $\bA\in\ZZ^{M\times M}$ such procedure can always ensure that $\bv_m\in\mathcal{C}_m$ for all $m$ for at least one mapping between codebooks and transmit antennas. Here, we combine this ingredient with the idea of using the decoded equations also for performing noise prediction. Namely, in the $m$th successive decoding step we compute
\begin{align}
\bY_{\text{eff}}^{(m)}&=\left[\bY_{\text{eff}}-\sum_{k=1}^{m-1}\left(\sqrt{\Tsnr} \ \mathbf{l}_k\bw_k-\bq_k\bv_k\right)\right]\Mod_c\nonumber\\
&=\left[\bV^{(m)}+\sqrt{\Tsnr}\bL^{(m)}\bW\right]\Mod_c,\nonumber
\end{align}
where the role of the column vectors $\{\mathbf{l}_k\}$ is to perform noise prediction, as before, and the role of the integer-valued column vectors $\{\bq_k\}$ is to reduce the number of lattice points participating in the remaining equations, such that only lattice points from $\mathcal{C}_m,\ldots,\mathcal{C}_M$ participate in $\bV^{(m)}$.

When doing so, however, one new issue arises. Reducing the rate of remaining equations using decoded ones is advantageous if the effective variances $\sigma^2_{\text{eff},1},\ldots,\sigma^2_{\text{eff},M}$ are monotonically increasing, such that the achievable computation rates are monotonically decreasing. Without noise prediction, one can always choose the decoding order such that this is satisfied, i.e., start with the best equation, then decode the second best and so on. When noise prediction is also applied, a situation that may occur is that after decoding the best equation and using it to predict the effective noise for the second equation, the second computation rate becomes higher than the first. This occurs if the diagonal entries of $\bL$, corresponding to the effective variances of the prediction errors, are not monotonically increasing. In this case, using decoded equations for reducing the rates of the remaining ones is less effective.

When the diagonal entries of $\bL$ are monotonically increasing, the described scheme can achieve any sum-rate satisfying
\begin{align}
\sum_{m=1}^M R_m&=-\frac{1}{2}\sum_{m=1}^M \log\left(\ell_{mm}^2\right)\nonumber\\
&=-\frac{1}{2}\log\left(\prod_{m=1}^M\ell_{mm}^2 \right)\nonumber\\
&=-\frac{1}{2}\log\det\left(\bL\bL^T\right)\nonumber\\
&=-\frac{1}{2}\log\det\left(\bA\left(\bI+\Tsnr\bH^T\bH\right)^{-1}\bA^T\right)\nonumber\\
&=\Cwi-\log|\det(\bA)|.\nonumber
\end{align}

The following definition from~\cite{CoFTransformFull} is needed for formally characterizing the performance of the described scheme.
\begin{definition}
For a full-rank $M\times M$ matrix $\bA$ with integer-valued entries we define the \emph{pseudo-triangularization} process,
which transforms the matrix $\bA$ to a matrix $\mathbf{\tilde{A}}$ which is upper triangular up to column permutation $\mathbf{\pi}=\left[\pi(1) \ \pi(2) \ \cdots  \ \pi(M)\right]$. This is accomplished by left-multiplying $\bA$ by a lower triangular matrix $\bR$ with unit diagonal, such that $\mathbf{\tilde{A}}=\bR\bA$ is upper triangular up to column permutation $\mathbf{\pi}$. Although the matrix $\bA$ is integer valued, the matrices $\bR$ and $\mathbf{\tilde{A}}$ need not necessarily be integer valued.
Note that the pseudo-triangularization process is reminiscent of Gaussian elimination except that row switching and row multiplication are prohibited.
\end{definition}

\begin{remark}
Any full-rank matrix can be triangularized using the Gaussian elimination process, and therefore any full-rank matrix can be pseudo-triangularized with at least one permutation vector $\mathbf{\pi}$.
\end{remark}
\vspace{1mm}

\begin{theorem}
\label{thm:SICnested}
Let $\bA\in\ZZ^{M\times M}$ be a full-rank target integer-valued matrix that can be pseudo-triangularized with the permutation vector $\mathbf{\pi}$, and let $\bL$ be the lower-triangular matrix from~\eqref{cholnp}, whose diagonal entries are $\ell_{ii}$. If $\ell_{11}^2\leq\cdots\leq\ell_{MM}^2$, then there exists a chain of nested lattice codebooks $\mathcal{C}_M\subseteq\cdots\subseteq\mathcal{C}_1$ with rates $R_M\leq\cdots\leq R_1$ such that if each $m$th antenna encodes its stream using the codebook $\mathcal{C}_{\mathbf{\pi}^{-1}(m)}$ with rate $R_{\mathbf{\pi}^{-1}(m)}$ and
\begin{align}
R_{\mathbf{\pi}^{-1}(m)}<-\frac{1}{2}\log(\ell_{\mathbf{\pi}^{-1}(m)\mathbf{\pi}^{-1}(m)}^2), \ \forall m=1,\ldots,M\nonumber%\label{ratePerCodebook}
\end{align}
all streams can be decoded with a vanishing error probability using the successive integer-forcing receiver. In other words, all streams can be decoded if the rate of each $m$th stream is smaller than $-\frac{1}{2}\log(\ell_{\mathbf{\pi}^{-1}(m)\mathbf{\pi}^{-1}(m)}^2)$. Consequently, the successive integer-forcing receiver can achieve any rate satisfying
\begin{align}
R_{\text{S-IF}}<\Cwi-\log|\det(\bA)|.\nonumber
\end{align}
\end{theorem}

\vspace{3mm}

In the described scheme each antenna transmits an independent stream. Thus, Theorem~\ref{thm:SICnested} remains valid for a Gaussian MAC with the $m$th column of $\bH$ representing the coefficients from the $m$th user to the receiver. In this case, the theorem implies that, if the stated conditions on $\bA$ are staisfied, there exists a chain of nested lattice codebooks and a mapping $\mathbf{\pi}^{-1}:\{1,\ldots,M\}\rightarrow\{1,\ldots,M\}$ between users and codebooks such that the $m$th user can achieve any rate below $-\frac{1}{2}\log(\ell_{\mathbf{\pi}^{-1}(m)\mathbf{\pi}^{-1}(m)}^2)$ with a vanishing error probability using the successive integer-forcing receiver.
Since $\Cwi$ is the MAC's sum-capacity, under the conditions of Theorem~\ref{thm:SICnested}, successive IF (which is usually termed successive compute-and-forward for a MAC) achieves the sum-capacity if $\bA$ is unimodular.

\vspace{2mm}

\begin{remark}
Note that for the choice $\bA=\bI$, successive IF corresponds to standard SIC. In this case, the monotonicity condition on the diagonal entries of $\bL$ is not needed. This is due to the fact that for the choice $\bA=\bI$ only one lattice point participates in each ``linear combination''. Therefore, the procedure from~\cite[Section IV.B]{CoFTransformFull}, which induces this monotonicity condition, is not needed.
\end{remark}

\vspace{2mm}

\begin{example}
\label{ex:mac}
Consider the two-user Gaussian MAC
\begin{align}
\by=\sqrt{2}\bx_1+\bx_2+\bz,\nonumber
\end{align}
at $\Tsnr=15$dB. For the choice $\bA=\bI$ and its row permutation, successive IF reduces to standard SIC, and results in the achievable rate-regions
\begin{align}
R_1<0.7776 \frac{\text{bits}}{\text{channel use}}, \ \ R_2<2.5139 \frac{\text{bits}}{\text{channel use}},\nonumber
\end{align}
and
\begin{align}
R_1<3.0028 \frac{\text{bits}}{\text{channel use}}, \ \ R_2<0.2887 \frac{\text{bits}}{\text{channel use}},\nonumber
\end{align}
respectively. For the choice
\begin{align}
\bA=\left(
      \begin{array}{cc}
        1 & 1  \\
        3 & 2  \\
      \end{array}
    \right)\nonumber
\end{align}
we have $\ell_{11}^2<\ell_{22}^2$, as Theorem~\ref{thm:SICnested} requires, and $-\nicefrac{1}{2}\log(\ell_{11}^2)=1.8452$ and $-\nicefrac{1}{2}\log(\ell_{22}^2)=1.4463$. In addition, $\bA$ can be pseudo-triangularized with the permutation vectors $\pi_1=[1 \ 2]$ and $\pi_2=[2 \ 1]$. It therefore follows that the two rate-regions
\begin{align}
R_1<1.8452 \frac{\text{bits}}{\text{channel use}}, \ \ R_2<1.4463 \frac{\text{bits}}{\text{channel use}},\nonumber
\end{align}
and
\begin{align}
R_1<1.4463 \frac{\text{bits}}{\text{channel use}}, \ \ R_2<1.8452 \frac{\text{bits}}{\text{channel use}},\nonumber
\end{align}
are achievable with successive IF. In addition, since $|\det(\bA)|=1$, these points are sum-rate optimal. Figure~\ref{fig:MACcapacity} shows the capacity region of the MAC from this example, along with the rate region achieved by successive IF.
\end{example}

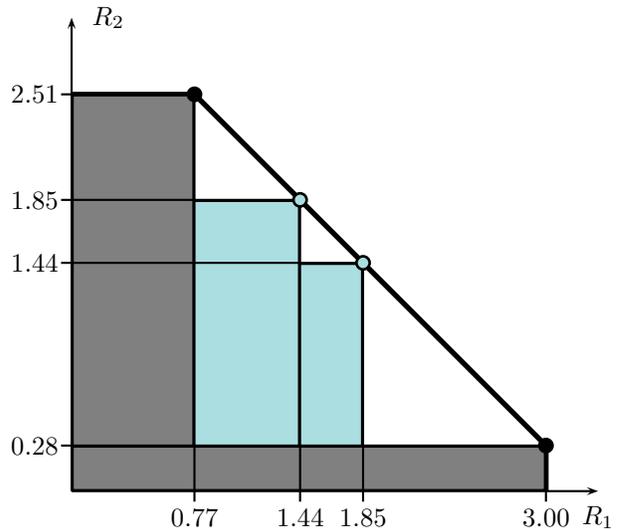
\begin{figure}[!htb]
\psset{unit=.7mm}
\begin{center}
\begin{pspicture}(-5,-5)(102,92)

\psline{->}(0,0)(100,0)
\psline{->}(0,0)(0,90)

\rput(7,90){$R_2$}
\psline{-}(-2,75.417)(2,75.417)
\rput(-7,75.417){$2.51$}
\psline{-}(-2,55.356)(2,55.356)
\rput(-7,55.356){$1.85$}
\psline{-}(-2,43.389)(2,43.389)
\rput(-7,43.389){$1.44$}
\psline{-}(-2,8.661)(2,8.661)
\rput(-7,8.661){$0.28$}

\rput(100,-5){$R_1$}
\psline{-}(90.084,-2)(90.084,2)
\rput(90.084,-5){$3.00$}
\psline{-}(55.356,-2)(55.356,2)
\rput(55.356,-5){$1.85$}
\psline{-}(43.389,-2)(43.389,2)
\rput(43.389,-5){$1.44$}
\psline{-}(23.328,-2)(23.328,2)
\rput(23.328,-5){$0.77$}

\psline[linewidth=2pt]{-}(0,75.417)(23.328,75.417)(90.084,8.661)(90.084,0)
\psframe[fillstyle=solid,fillcolor=paleblue](0,0)(55.356,43.389)
\psframe[fillstyle=solid,fillcolor=paleblue](0,0)(43.389,55.356)
\psframe[fillstyle=solid,fillcolor=gray](0,0)(23.328,75.417)
\psframe[fillstyle=solid,fillcolor=gray](0,0)(90.084,8.661)

\psline(0,75.417)(23.328,75.417)(23.328,0)
\psline(0,8.661)(90.084,8.661)(90.084,0)
\psline(0,43.389)(55.356,43.389)(55.356,0)
\psline(0,55.356)(43.389,55.356)(43.389,0)

\pscircle[linewidth=1pt,fillcolor=black,fillstyle=solid](23.328,75.417){1.5}
\pscircle[linewidth=1pt,fillcolor=black,fillstyle=solid](90.084,8.661){1.5}

\pscircle[linewidth=1pt,fillcolor=paleblue,fillstyle=solid](55.356,43.389){1.5}
\pscircle[linewidth=1pt,fillcolor=paleblue,fillstyle=solid](43.389,55.356){1.5}

\end{pspicture}
\end{center}
\caption{The capacity region of the MAC from Example~\ref{ex:mac}, and the rate region achieved by successive IF. The gray (dark shaded) area in the figure corresponds to the rate-region achievable by standard SIC, whereas the pale blue (bright shaded) area is the additional rate-region obtained by successive IF. Note that the plotted rate-region does not include time-sharing.}\label{fig:MACcapacity}
\end{figure}

\vspace{2mm}

\begin{remark}
An interesting conclusion from Theorem~\ref{thm:SICnested} is that if $\bA$ is such that the conditions of the theorem are satisfied, then the compute-and-forward decoder used in~\cite{ng11IT} achieves the same rates as the optimal maximum-likelihood (ML) decoder for decoding the integer-linear combinations whose coefficients are the rows of $\bA$. Recall that the decoder from~\cite{ng11IT} scales its observations, subtracts the dithers, quantizes to the fine lattice and reduces the result modulo the coarse lattice. This is in contrast to the ML decoder that computes the likelihood of each possible outcome for the desired integer-linear combination given the channel's output, and chooses the one that is most likely. The achievable rates described by Theorem~\ref{thm:SICnested} are attained using the same decoder as in~\cite{ng11IT} at each decoding step (after noise prediction was applied). It is shown that the obtained sum-rate equals the MAC's sum-capacity. If the ML decoder could attain higher rates, than using the noise prediction scheme described in this section, with the ML decoder instead of the one from~\cite{ng11IT}, higher rates could be obtained, contradicting the converse theorem for the MAC capacity region.
We emphasize that this only shows that the decoder of~\cite{ng11IT} achieves the highest possible rate for \emph{the described transmission scheme}, where each user transmits a dithered version of a lattice point taken from a chain of nested lattice codebook. Our results do not preclude the possibility that integer-linear combinations can be reliably decoded with a higher computation rate than that of~\cite{ng11IT} under \emph{a different transmission scheme}.
\end{remark}

\subsection{From  \hspace{-0.13in} $\mod \hspace{-0.05in}\mbox{-}\Lambda$ decoding to decoding over the reals}
\label{subsec:correctzeff}
Successive IF uses the values of the effective noise vectors $\bz_{\text{eff},1},\ldots,\bz_{\text{eff},m-1}$ in order to estimate $\bz_{\text{eff},m}$ and reduce its effective variance. However, correct decoding of the equation $\bv_k$ from $\by_{\text{eff},k}=[\bv_k+\bz_{\text{eff},k}]\Mod_c$ only ensures that the receiver has access to $[\bz_{\text{eff},k}]\Mod_c$, whereas for the successive IF scheme $\bz_{\text{eff},k}$ is needed.
Correct decoding of $\bv_k$ from $\by_{\text{eff},k}$ only implies that $[\bz_{\text{eff},k}]\Mod_c\in\CV_f$, but does not necessarily imply that $\bz_{\text{eff},k}\in\CV_f$. If the fine lattice $\Lambda_f$ used for constructing $\mathcal{C}$ is Poltyrev-good and $\Lambda_C$ is Rogers-good~\cite{ez04}, we have
\begin{align}
\Pr\left(\bz_{\text{eff},k}\notin\CV_f\right)<2^{-n\left(\frac{1}{2}\log\left(\frac{\Tsnr}{\sigma^2_{\text{eff},k}}\right)-R-o_1(n)\right) },\label{poltyrev}
\end{align}
and since $\CV_f\subset\CV_c$ the probability that $[\bz_{\text{eff},k}]\Mod_c\neq\bz_{\text{eff},k}$ can be made arbitrarily small by increasing the block length $n$. In~\eqref{w1hat} we assumed that $\Lambda_f$ is Poltyrev-good and $\Lambda_c$ is Rogers-good to obtain the relation $[\bz_{\text{eff},k}]\Mod_c\stackrel{w.h.p.}{=}\bz_{\text{eff},k}$.

In practice, however, a ``good'' nested lattice codebook is hard to implement, and suboptimal nested lattice codebooks are constructed. A commonly used construction for a nested lattice codebook is one where the fine lattice is built from a linear code of block length $n$ over a prime field $\ZZ_p$ with moderate cardinality (e.g., an LDPC code or a turbo code) using Construction A~\cite{cs88,loeliger97}, and the coarse lattice is the scaled integer-lattice $p\ZZ^n$. In this case, although $\Pr([\bz_{\text{eff},k}]\Mod_c\notin\CV_f)$ can be made as small as desired for $n$ large enough and an appropriate choice of $R$, $\Pr([\bz_{\text{eff},k}]\Mod_c\neq\bz_{\text{eff},k})$ cannot be made arbitrarily small, as $\Lambda_c$ is the scaled integer lattice whose Euclidean minimum distance does not increase with $n$. In other words, $\Lambda_f$ is not Poltyrev good. Thus, for construction A nested lattice codebooks, correct decoding of $\bv_k$ does not ensure correct decoding of $\bz_{\text{eff},k}$. This may degrade the performance of successive IF, as the predictions of subsequent effective noises may also be impaired.

Nevertheless, we claim that for moderate (not too small) rates of the codebook $\mathcal{C}$, this type of error will have a negligible effect on the total error probability. It can be shown that for a Construction A nested lattice codebook, if $R$ is such that $\bv_k$ can be decoded reliably from $\by_{\text{eff},k}$, then
\begin{align}
\Pr\left([\bz_{\text{eff},k}(i)]\hspace{-0.10in}\mod p\neq\bz_{\text{eff},k}(i)\right)<\exp\left\{-\frac{\pi e}{4}2^{2R} \right\}\nonumber
\end{align}
for each of the $n$ components of $\bz_{\text{eff},k}$. Thus, the expected number of components where $[\bz_{\text{eff},k}]\hspace{-0.08in}\mod p\neq\bz_{\text{eff},k}$ is fairly small for moderate $R$, and these erroneous components will not degrade the performance of successive IF by much.
\section{Finding the Optimal Integer-valued matrix $\bA$}
\label{sec:optA}
 Thus far, we have described the successive IF scheme for some predefined integer-valued matrix $\bA$. The performance of IF, as well as successive IF, critically depends on the choice of $\bA$ and we now discuss a procedure for finding its optimal value.

We would like to find the matrix $\bA$ that maximizes the computation rate for the worst equation when noise prediction is used. Mathematically, this problem can be formulated as
\begin{align}
\bA^{\text{opt}}_{\text{S-IF}}&=\argmin_{\substack{{\bA\in\ZZ^{M\times M}}\\ {\det(\bA)\neq 0}}} \ \max_{k=1,\ldots,M}\ell_{kk}^2,\label{SIFopt}
\end{align}
where $\ell_{kk}$ are the diagonal entries of the lower triangular matrix $\bL$ defined in~\eqref{cholnp}.
Note that this problem is different than the optimization problem for standard IF, which can be written as
\begin{align}
\bA^{\text{opt}}_{\text{IF}}&=\argmin_{\substack{{\bA\in\ZZ^{M\times M}}\\ {\det(\bA)\neq 0}}} \ \max_{k=1,\ldots,M}\sum_{i=1}^k\ell_{ik}^2.\label{IFopt}
\end{align}
As a result, the solution of~\eqref{SIFopt} may be different that that of~\eqref{IFopt}. We now show that for successive IF we can restrict $\bA$ to the class of unimodular matrices (matrices with integer entries and determinant $\pm 1$) without loss of generality, and its optimal value can be obtained using the Korkin-Zolotarev basis reduction procedure.

\vspace{1mm}

\begin{definition}[Korkin-Zolotarev basis~\cite{mg02}]
\label{def:KZ}
Let $\bF=[\bff_1 \ \cdots \ \bff_M]$ be a lattice basis of rank $M$, and let $\bF^*=[\bff_1^* \ \cdots \ \bff_M^*]$ be its corresponding Gram-Schmidt orthogonalization, i.e., $\bF=\bF^*\cdot\bR$ for some upper triangular matrix $\bR$ with unit diagonal. Define the projection functions $\mathcal{P}_i(\bx)=\sum_{j\geq i}(\bx^T \bff_j^* / \|\bff_j^*\|^2)\bff_j^*$ that project $\bx$ onto $\Span(\bff_i^*,\ldots,\bff_M^*)$. The basis $\bF$ is Korkin-Zolotarev (KZ) reduced if and only if for all $i=1,\ldots,M$
\begin{itemize}
\item $\bff^*_i$ is a shortest nonzero vector in $\mathcal{P}_i(\Lambda(\bF))$
\item for all $j>i$, the Gram-Schmidt coefficients $r_{j,i}=\bff_j^T\bff_i^*/\|\bff_i^*\|^2$ of $\bF$ satisfy $|r_{j,i}|\leq 1/2$.
\end{itemize}
\end{definition}

\vspace{1mm}

\begin{theorem}
\label{thm:KZ}
Let $\bG$ be defined as in~\eqref{cholnoA} and let $\bA$ be a unimodular matrix. If $\bG^T\bA^T$ is a KZ basis of the lattice $\Lambda(\bG^T)\triangleq\{\bG^T\bx \ : \  \bx\in\ZZ^M\}$ then $\bA$ is an optimal integer-valued matrix for successive IF.
\end{theorem}

\vspace{1mm}

\begin{proof}
Let $\bA$ be a unimodular matrix such that $\bG^T\bA^T$ is a KZ basis of the lattice $\Lambda(\bG^T)$. Such a matrix always exists~\cite{mg02}. Let $\ba_i^T$ be the $i$th row of $\bA$, and let $\bL$ be the lower triangular matrix defined in~\eqref{cholnp}. Note that $\ell_{ii}$ depends only on $\{\ba_1,\ldots,\ba_{i}\}$, and is independent of $\{\ba_{i+1},\ldots,\ba_{M}\}$.

We first show that out of all integer-valued vectors that are linearly independent of $\{\ba_1,\ldots,\ba_{i-1}\}$, $\ba_i$ yields the minimum value of $\ell^2_{ii}$. Let $\bF=[\bff_1 \ \cdots \ \bff_M]=\bG^T\bA^T$ and let $\bF^*=[\bff_1^* \ \cdots \ \bff_M^*]$ be the corresponding Gram-Schmidt orthogonalized basis, such that $\bF=\bF^*\cdot\bR$ for an upper triangular matrix $\bR$ with unit diagonal. We further define the unitary matrix $\bU\triangleq\bF^*\cdot\diag(\|\bff^*_1\|^{-1},\ldots,\|\bff^*_M\|^{-1})$ and the upper triangular matrix $\bL^T\triangleq \diag(\|\bff^*_1\|,\ldots,\|\bff^*_M\|)\cdot \bR$, such that $\bF=\bU\bL^T$ and
\begin{align}
\bF^T\bF=\bA\bG\bG^T\bA^T=\bL\bL^T.\nonumber
\end{align}
From the uniqueness of the Cholesky decomposition, it follows that the $\bL$ defined above is the same as in~\eqref{cholnp}.
Define the projection functions $\mathcal{P}_i(\bx)=\sum_{j\geq i}(\bx^T \bff_j^* / \|\bff_j^*\|^2)\bff_j^*$ that project $\bx$ onto $\Span(\bff_i^*,\ldots,\bff_M^*)$. We have
\begin{align}
\ell^2_{ii}=\|\bff_i^* \|^2=\|\mathcal{P}_{i}(\bff_i)\|^2=\|\mathcal{P}_{i}(\bG^T\ba_i)\|^2.\label{lii}
\end{align}
By definition of the KZ reduction, $\mathcal{P}_{i}(\bG^T\ba_i)$ is a shortest nonzero vector in $\mathcal{P}_{i}(\Lambda(\bG^T))$, which means that
\begin{align}
\ba_i=\argmin_{\substack{{\ba\in\ZZ^{M}}\\ {\rank(\ba_1,\ldots,\ba_{i-1},\ba)=i}}}\|\mathcal{P}_{i}(\bG^T\ba)\|^2=\argmin_{\substack{{\ba\in\ZZ^{M}}\\ {\rank(\ba_1,\ldots,\ba_{i-1},\ba)=i}}}\ell^2_{ii},\label{aieq}
\end{align}
as desired, where the last equality follows from~\eqref{lii}.

To establish the optimality of $\bA$ for successive IF, it remains to show that a greedy procedure that for each $i$ selects the $\ba_i$ as in~\eqref{aieq} also minimizes the value of $\max_{k=1,\ldots,M}\ell^2_{kk}$. Let $\tilde{\bA}=[\tilde{\ba}_1 \ \cdots \ \tilde{\ba}_M]$ be a ``competing'' full-rank matrix with integer-valued entries. Define the matrix $\tilde{\bF}=[\tilde{\bff}_1 \ \cdots \ \tilde{\bff}_M]=\bG^T\tilde{\bA}$ and define the matrix $\tilde{\bL}$ as the lower triangular matrix in the Cholesky decomposition of $\tilde{\bF}^T\tilde{\bF}$, i.e., $\tilde{\bF}^T\tilde{\bF}=\tilde{\bL}\tilde{\bL}^T $. For the choice $\tilde{\bA}$, the achievable rate for successive IF is dictated by $\max_{k=1,\ldots,M}\tilde{\ell}^2_{kk}$, where $\tilde{\ell}_{kk}$ are the diagonal entries of $\tilde{\bL}$. Let $\tilde{\bF}^*=[\tilde{\bff}_1^* \ \cdots \ \tilde{\bff}_M^*]$ be the Gram-Schmidt orthogonalized basis corresponding to $\tilde{\bF}$ and define the projection functions $\tilde{\mathcal{P}}_i(\bx)=\sum_{j\geq i}(\bx^T \tilde{\bff}_j^* / \|\tilde{\bff}_j^*\|^2)\tilde{\bff}_j^*$ that project $\bx$ onto $\Span(\tilde{\bff}_i^*,\ldots,\tilde{\bff}_M^*)$. Note that we have $\tilde{\ell}^2_{ii}=\|\tilde{\mathcal{P}}_i(\bG^T\tilde{\ba}_i)\|^2$.

In order to prove that $\bA$ is optimal, we show by induction that for each $m=1,\ldots,M$
\begin{align}
\max_{k=1,\ldots,m}\ell^2_{kk}\leq \max_{k=1,\ldots,m}\tilde{\ell}^2_{kk}.\label{induction}
\end{align}
The induction hypothesis~\eqref{induction} holds for $m=1$ since $\ell^2_{11}=\|\bff_1\|^2$, and by definition of the KZ reduction $\bff_1$ is a shortest vector in $\Lambda(\bG^T)$. We assume~\eqref{induction} holds for $m-1$ and show that it also holds for $m$. We have
\begin{align}
\max_{k=1,\ldots,m}\ell^2_{kk}&=\max\left(\ell^2_{mm},\max_{k=1,\ldots,m-1}\ell^2_{kk} \right)\nonumber\\
&\leq \max\left(\ell^2_{mm},\max_{k=1,\ldots,m-1}\tilde{\ell}^2_{kk} \right)\label{indhyp}\\
&=\max\left(\|\mathcal{P}_m(\bG^T\ba_m)\|^2,\max_{k=1,\ldots,m-1}\tilde{\ell}^2_{kk} \right)\label{maxineq}
\end{align}
where~\eqref{indhyp} follows from the induction hypothesis.

If $\Span(\ba_1,\ldots,\ba_{m-1})=\Span(\tilde{\ba}_1,\ldots,\tilde{\ba}_{m-1})$ we have $\|\mathcal{P}_{m}(\bx)\|=\|\tilde{\mathcal{P}}_m(\bx)\|$ for any $\bx\in\RR^M$. Therefore,
\begin{align}
\|\mathcal{P}_m(\bG^T\ba_m)\|^2\leq\|{P}_m(\bG^T\tilde{\ba}_m)\|^2=\|\tilde{\mathcal{P}}_m(\bG^T\tilde{\ba}_m)\|^2=\tilde{\ell}^2_{mm},\label{samespan}
\end{align}
where the first inequality follows from the definition of the KZ reduction. Substituting~\eqref{samespan} into~\eqref{maxineq} gives~\eqref{induction}.

If $\Span(\ba_1,\ldots,\ba_{m-1})\neq\Span(\tilde{\ba}_1,\ldots,\tilde{\ba}_{m-1})$, let $j$ be the smallest index for which $\tilde{\ba}_j\notin \Span(\ba_1,\ldots,\ba_{m-1})$. It follows that $\Span(\tilde{\ba}_1,\ldots,\tilde{\ba}_{j-1})\subset\Span(\ba_1,\ldots,\ba_{m-1})$ and therefore $\|\mathcal{P}_{m}(\bx)\|\leq \|\tilde{\mathcal{P}}_{j}(\bx)\|$ for any $\bx\in\RR^M$.
We have
\begin{align}
\|\mathcal{P}_m (\bG^T\ba_m)\|^2&=\min_{\substack{{\ba\in\ZZ^{m}}\\{\rank(\ba_1,\ldots,\ba_{m-1},\ba)=m}}}\|\mathcal{P}_{m}(\bG^T\ba)\|^2 \label{eq1}\\
&\leq \|\mathcal{P}_{m}(\bG^T\tilde{\ba}_j)\|^2\label{ineq2}\\
&\leq \|\tilde{\mathcal{P}}_{j}(\bG^T\tilde{\ba}_j)\|^2\label{ineq3}\\
&=\tilde{\ell}^2_{jj},\label{difspan}
\end{align}
where~\eqref{eq1} follows from~\eqref{aieq}, \eqref{ineq2} follows since $\tilde{\ba}_j\notin \Span(\ba_1,\ldots,\ba_{m-1})$ and is therefore included in the minimization space, and \eqref{ineq3} follows since $\|\mathcal{P}_{m}(\bx)\|\leq \|\tilde{\mathcal{P}}_{j}(\bx)\|$ for any $\bx\in\RR^M$.
Substituting~\eqref{difspan} into~\eqref{maxineq} gives~\eqref{induction}.
\end{proof}

\vspace{1mm}

It is well known~\cite{mg02}, and not too difficult to verify, that if a linearly independent set of lattice vectors $\bS=[\bs_1 \ \cdots \ \bs_M]$ is KZ reduced, then $\bS$ is a basis for the original lattice. For this reason, in contrast to standard IF where the optimal $\bA$ is not necessarily unimodular, for successive IF there is no loss (in terms of achievable rate) in restricting $\bA$ to the class of unimodular matrices. It follows that for uncoded PAM transmission (or equivalently using the $1-D$ integer lattice as codebook), successive IF and lattice-reduction (LR) aided SIC~\cite{wbkk04} are in fact equivalent. Although the advantages of the KZ reduction for lattice-reduction-aided SIC were pointed out in the literature~\cite{ling11}, to the best of our knowledge, there is no prior work on its optimality in terms of minimizing the error probability.

Finding a KZ basis for a lattice is known to be NP-hard in general, as it involves finding a shortest lattice vector, which is itself NP-hard. The following is a recursive procedure for finding a KZ basis $\bF=[\bff_1 \ \cdots \ \bff_M]$ for a rank $M$ lattice $\Lambda$. Let $\bff_1$ be a shortest vector in $\Lambda$, and let $\Lambda'$ be the lattice given by the orthogonal projection of $\Lambda$ on the subspace of $\Span(\Lambda)$ orthogonal to $\bff_1$ (it can be verified that $\Lambda'$ is indeed a lattice). Let $\bc_2,\ldots,\bc_M$ be the KZ basis of $\Lambda'$. Define $\bff_i=\bc_i+\alpha_i\bff_1$, where $\alpha_i\in(-1/2,1/2]$ is the unique number such that $\bff_i\in\Lambda$, for $i=2,\ldots,M$.

For channels of small dimensions the KZ basis can be computed exactly. For large dimensions, it can be approximated by applying the LLL algorithm $M$ successive times, where the dimension of the lattice for which LLL is applied decreases at each iteration. Such as algorithm is described in~\cite[Section VI.D]{aevz02}, and Matlab code for finding an approximation for the optimal $\bA$ based on this method can be found in~\cite{AsicCode}.

\begin{appendices}

\section{Successive IF via MMSE-GDFE}
\label{app:GDFE}
In Section~\ref{sec:SIF}, we described the implementation of successive IF via noise-prediction. In this appendix, we show an equivalent implementation of successive IF where the decoded equations themselves, instead of the effective noises, are used for improving the achievable rates for decoding subsequent equations. As for successive IF via noise-prediction, the derivation relies on linear MMSE estimation theory, and in particular on the MMSE-GDFE framework.

A key fact used in the derivation is that if a certain equation $\bv=[\ba^T\bT]\Mod_c$ can be decoded from $\bY$, then with high probability $\ba^T\bX$ can also be recovered from $\bY$. This fact is proved in~\cite[Lemma 1]{nazer12IZS}, and follows from the same considerations discussed in Section~\ref{subsec:correctzeff}. Therefore, when attempting to decode the equation $\bv_k=[\ba_k^T\bT]\Mod_c$, the receiver already has access to $\ba_1^T\bX,\ldots,\ba_{k-1}^T\bX$.
Thus, for any lower-triangular matrix $\bC$ with diagonal entries equal to zero, the successive IF receiver can produce the effective channel
\begin{align}
\bY_{\text{eff}}&=\left[\bB\bY-\bC\bA\bX+\bA\bD\right]\Mod_c\nonumber\\
&=\left[\bA\bT+\left(\bB\bH-\bA-\bC\bA\right)\bX+\bB\bZ\right]\Mod_c\nonumber\\
&=\left[\bV+\bE\right]\Mod_c,
\end{align}
where
\begin{align}
\bE\triangleq\left(\bB\bH-\bR\bA\right)\bX+\bB\bZ
\end{align}
is statistically independent of $\bT$, and $\bR\triangleq\bI+\bC$.
Note that the constrained structure of $\bC$ ensures that in all steps of the successive decoding procedure the receiver only uses values of $\ba_i^T\bX$ that are already available to it. The filter matrix $\bR$ is monic, i.e., a lower-triangular matrix with unit diagonal entries, and together with $\bB$ can be optimized such as to minimize the generalized covariance matrix $\mathbf{\tilde{K}}_{\bE\bE}$. For a given choice of $\bR$, the filter $\bB$ should be chosen as the optimal linear MMSE estimation filter of $\bR\bA\bX$ from $\bY$, which is given by
\begin{align}
\bB=\bR\bA{\bH}^T\left(\frac{1}{\Tsnr}\bI+{\bH}{\bH}^T\right)^{-1},\nonumber
\end{align}
and the resulting estimation-error generalized covariance matrix is
\begin{align}
\mathbf{\tilde{K}}_{\bE\bE}=\Tsnr\cdot(\bR\bA)\left(\bI+\Tsnr{\bH}^T{\bH}\right)^{-1}(\bR\bA)^T.\label{covgdfe}
\end{align}
Comparing~\eqref{covgdfe} to the estimation error covariance matrix obtained in standard IF
\begin{align}
\Tsnr\cdot\bA\left(\bI+\Tsnr{\bH}^T{\bH}\right)^{-1}\bA^T\nonumber
\end{align}
reveals the advantage of successive IF over standard IF. It essentially allows to decode any full rank-set of equations of the form $\bR\bA\bX$, where $\bR$ is some monic filter, rather than just equations of the form $\bA\bX$.

Recall that the performance of the IF receiver are dictated by the effective variances of the effective noises, i.e., by the diagonal entries of $\mathbf{\tilde{K}}_{\bE\bE}$. Thus, for a given choice of $\bA$, $\bR$ should be chosen such as to minimize the values of these entries. Note that,
\begin{align}
\mathbf{\tilde{K}}_{\bE\bE}=\Tsnr\cdot(\bR\bL)(\bR\bL)^T,\nonumber
\end{align}
where $\bL$ is the lower triangular matrix defined in~\eqref{cholnp}. The $i$th diagonal entry of $\mathbf{\tilde{K}}_{\bE\bE}$ is therefore equal to $\Tsnr$ times the squared Euclidean norm of the $i$th row in the matrix $\bR\bL$. Now, since $\bR$ is monic, we must have $\{\bR\bL\}_{ii}=\ell_{ii}$, which implies that
\begin{align}
\left\{\mathbf{\tilde{K}}_{\bE\bE}\right\}_{ii}\geq \Tsnr\ell_{ii}^2, \ \forall i=1,\ldots,M.\label{gdfevarbound}
\end{align}
The choice
\begin{align}
\bR=\diag(\ell_{11},\ldots,\ell_{MM})\bL^{-1},
\end{align}
attains the bound from~\eqref{gdfevarbound}, and is therefore optimal. The resulting estimation-error generalized covariance matrix is
\begin{align}
\mathbf{\tilde{K}}_{\bE\bE}=\Tsnr\cdot\diag(\ell^2_{11},\ldots,\ell^2_{MM}),\nonumber
\end{align}
and as a result the achievable rates are exactly as in Theorem~\ref{thm:SIC}, which shows that successive IF via noise-prediction or via MMSE-GDFE are indeed equivalent.
\end{appendices}

\bibliographystyle{IEEEtran}
\bibliography{IF_SIC_bib}

\end{document}